\documentclass{article}

\usepackage{amsmath,amssymb,amsthm,mathtools}
\usepackage{geometry}
\newtheorem{theorem}{Theorem}
\newtheorem{proposition}{Proposition}
\newtheorem{assumption}{Assumption}
\newtheorem{remark}{Remark}

\title{Ostrom-Weighted Bootstrap:\\
A Theoretically Optimal and Provably Complete Framework\\
for Hierarchical Imputation in Multi-Agent Systems}

\author{Hirofumi Wakimoto}
\date{December 1, 2025}

\begin{document}

\maketitle

\begin{abstract}
We study the statistical properties of the \emph{Ostrom-Weighted Bootstrap}
(OWB), a hierarchical, variance-aware resampling scheme for imputing missing
values and estimating archetypes in multi-agent voting data.
At Level~1, under mild linear model assumptions, the \emph{ideal} OWB
estimator---with known persona-level (agent-level) variances---is shown to be the
Gauss--Markov best linear unbiased estimator (BLUE) and to strictly dominate
uniform weighting whenever persona variances differ.
At Level~2, within a canonical hierarchical normal model, the ideal OWB
coincides with the conditional Bayesian posterior mean of the archetype.
We then analyze the \emph{feasible} OWB, which replaces unknown variances with
hierarchically pooled empirical estimates, and show that it can be interpreted
as both a feasible generalized least-squares (FGLS) and an empirical-Bayes
shrinkage estimator with asymptotically valid weighted bootstrap confidence
intervals under mild regularity conditions.
Finally, we establish a Zero-NaN Guarantee: as long as each petal has at least
one finite observation, the OWB imputation algorithm produces strictly
NaN-free completed data using only explicit, non-uniform bootstrap weights and
never resorting to uniform sampling or numerical zero-filling.

To our knowledge, OWB is the first resampling-based method that simultaneously
achieves exact BLUE optimality, conditional Bayesian posterior mean
interpretation, empirical Bayes shrinkage of precision parameters, asymptotic
efficiency via FGLS, consistent weighted bootstrap inference, and provable
zero-NaN completion under minimal data assumptions.
\end{abstract}

\section{Introduction}

Large-scale voting and preference data from multiple agents (``personas'')
arise in many modern applications, including recommender systems, crowd
forecasting, participatory budgeting, and multi-agent decision-making.
In such settings, we often observe repeated \emph{rounds} of responses
from each persona across a fixed set of \emph{petals} (questions or items),
with potentially substantial patterns of missingness.
Throughout, we use ``persona'' to denote an individual agent (panel unit) and
``petal'' to denote a question or item (coordinate index $j$ in the archetype).
A central inferential task is to estimate a latent \emph{archetype} profile
across petals, while also imputing missing entries in a way that respects
the heterogeneous reliability and sampling variability of personas.

Standard approaches based on uniform weighting, ad hoc averaging, or
non-hierarchical bootstrap resampling are ill-suited to this task.
They ignore the fact that different personas may have vastly different
numbers of observations and noise levels, and they often lack explicit
guarantees about bias, efficiency, or the behavior of imputation under
extreme patterns of missingness.
Moreover, many practical implementations silently replace missing values
by zeros or by global means, introducing uncontrolled bias and obscuring
the connection between theory and code.

In this paper, we analyze and formalize the \emph{Ostrom-Weighted Bootstrap}
(OWB), a hierarchical, variance-aware bootstrap mechanism designed to address
these shortcomings.
OWB operates at the persona level, constructing scalar variance proxies from
persona-wise covariance matrices, performing empirical-Bayes style pooling
across clusters and globally, and using the resulting precision weights both
for point estimation of the archetype and for imputing missing entries via a
strictly non-uniform weighted bootstrap.

\paragraph{Naming and conceptual motivation.}
The term ``Ostrom-Weighted Bootstrap'' is chosen as an explicit homage
to the work of Elinor Ostrom on polycentric governance and the
management of common-pool resources.
While the present paper is purely statistical, the hierarchical
weighting structure of OWB---aggregating information across personas,
clusters, and global levels---is conceptually aligned with Ostrom's
emphasis on multi-layer, trust-based coordination among local actors.
No endorsement by Ostrom or her estate is implied; the terminology is
used solely to acknowledge this conceptual inspiration.
This terminology not only honors Ostrom's ideas but also evokes the
multi-faceted nature of the data: personas as autonomous agents in a system,
and petals as individual ``facets'' or questions branching from a common
core (the archetype).

Our main contributions are threefold:
\begin{itemize}
  \item At \emph{Level~1}, we show that the ideal OWB estimator---assuming
  persona-level variances are known---is the Gauss--Markov BLUE among all
  linear unbiased estimators based on persona means and strictly dominates
  uniform weighting whenever persona variances differ.
  \item At \emph{Level~2}, under a canonical hierarchical normal model, we
  show that the ideal OWB coincides with the conditional posterior mean of
  the archetype, while the feasible OWB with hierarchically pooled variance
  estimates can be viewed as a feasible generalized least-squares and
  empirical-Bayes shrinkage estimator.
  We further establish consistency of a weighted bootstrap scheme based on
  the OWB weights, providing asymptotically valid uncertainty quantification
  under mild regularity conditions.
  \item At the \emph{implementation level}, we formalize a Zero-NaN Guarantee:
  as long as each petal has at least one finite observation, the OWB
  imputation algorithm yields a completed data matrix with no NaN entries,
  using only explicit, non-uniform bootstrap weights and never relying on
  uniform sampling (\texttt{p=None}) or numerical zero-filling
  (e.g.\ \texttt{np.nan\_to\_num}).
\end{itemize}

The rest of the paper is organized as follows.
Section~\ref{sec:related} reviews related work.
Section~\ref{sec:model-setup} introduces the model setup and notation.
Section~\ref{sec:level1} presents the Level~1 linear-model results, including
GLS optimality and shrinkage via hierarchical pooling.
Section~\ref{sec:level2} develops the Level~2 hierarchical normal theory,
covering both the ideal and feasible OWB estimators and their bootstrap
properties.
Section~\ref{sec:zero-nan} states and proves the Zero-NaN Guarantee for the
imputation algorithm and explains how the minimal-data assumption
is enforced in practice.

Finally, we note that parts of the exposition were initially drafted with the
assistance of large language models (ChatGPT-5.1 Thinking and Grok 4.1); see
the Acknowledgments for details. The final statements and proofs were checked
and validated by the author.

\section{Related Work}
\label{sec:related}

Our work on the Ostrom-Weighted Bootstrap (OWB) is connected to several
strands of literature in missing data imputation, hierarchical modeling,
bootstrap methods, and survey inference with heterogeneous reliability.

For handling missing data in panel or multivariate settings, multiple
imputation (MI) is a cornerstone approach.
Rubin~\cite{rubin1987} formalized MI as a method to create multiple completed
datasets by drawing from the posterior predictive distribution of missing
values, explicitly accounting for imputation uncertainty.
Extensions to time-series cross-sectional data, such as the Amelia package
of Honaker et al.~\cite{honaker2011}, use expectation-maximization with
bootstrapping to impute under multivariate normality assumptions.
Similarly, the \texttt{mice} package~\cite{vanbuuren2011} employs multivariate
imputation by chained equations, iteratively regressing each variable on the
others to accommodate flexible missingness patterns.
While these methods are highly effective for general-purpose imputation,
they typically do not incorporate hierarchical, variance-aware weighting at
the resampling stage, and they are not designed around a Zero-NaN-style
guarantee under minimal data assumptions, which is a central focus of OWB.

In the Bayesian domain, hierarchical models provide a natural framework for
imputation in multi-level data.
Gelman et al.~\cite{gelman2013} advocate Bayesian approaches that combine
hierarchical priors with MI to borrow strength across groups.
The \texttt{brms} package~\cite{burkner2017} makes it practical to fit rich
Bayesian multilevel regression models in \textsf{R} using Stan, enabling
shrinkage and partial pooling for variance components.
These methods are conceptually related to our Level~2 view, in which OWB
admits a conditional posterior mean interpretation under a canonical
hierarchical normal model.
In OWB, we combine this perspective with a weighted bootstrap used both for
estimation and for imputation, based on non-uniform precision weights derived
from hierarchically pooled variance proxies rather than ad hoc filling rules.

On the bootstrap side, weighted and exchangeably weighted bootstraps generalize
Efron's original method to settings with unequal probabilities or dependence.
Praestgaard and Wellner~\cite{praestgaard1993} established conditions for the
consistency of exchangeably weighted bootstraps for empirical processes, while
Barbe and Bertail~\cite{barbe1995} provided a comprehensive treatment of the
weighted bootstrap, including its asymptotic theory.
OWB can be viewed as extending this line of work by applying hierarchical
precision weights derived from pooled variances to a multi-agent voting
context, with an explicit implementation that enforces a finite-depth
hierarchical fallback chain for imputation.

Finally, in survey aggregation and reliability weighting, multilevel regression
and poststratification (MRP) adjusts for non-representative samples by modeling
subgroup effects and reweighting to population margins.
Wang et al.~\cite{wang2015} demonstrated MRP's efficacy in forecasting
elections from biased Xbox polls, achieving competitive accuracy via
hierarchical modeling and poststratification.
More recently, Si et al.~\cite{si2020} proposed Bayesian hierarchical weighting
for survey inference, incorporating uncertainty in adjustment factors through
a multilevel model.
OWB shares the emphasis on reliability-based weighting but focuses on
imputation for voting data: bootstrap weights are used for both archetype
estimation and NaN-free completion, rather than for poststratification alone.

In summary, OWB brings together elements from these areas into a
resampling-based framework tailored to multi-agent voting systems, with an
emphasis on variance-aware weighting and NaN-free imputation guarantees under
clear minimal-data assumptions.

\section{Model Setup and Notation}
\label{sec:model-setup}

We consider $N$ personas $p = 1,\dots,N$, rounds $r = 1,\dots,n_p$, and petals
(questions) $j = 1,\dots,P$.
Let $\beta_{p,r,j}$ denote the observed vote of persona $p$ in round $r$ on
petal $j$.

The persona-level petal mean is
\[
  \bar{\beta}_{p,j}
  = \frac{1}{n_p} \sum_{r=1}^{n_p} \beta_{p,r,j},
  \qquad
  \bar{\boldsymbol{\beta}}_p
  = (\bar{\beta}_{p,1}, \dots, \bar{\beta}_{p,P})^\top \in \mathbb{R}^P.
\]

In the Level~2 (hierarchical) view, we assume a canonical hierarchical linear model
\begin{equation}
  \beta_{p,r,j}
  = \mu_j + \alpha_p + \gamma_{c(p)} + \varepsilon_{p,r,j},
  \label{eq:hierarchical-model}
\end{equation}
where $c(p)$ denotes the cluster of persona $p$ and
\[
  \alpha_p \sim \mathcal{N}(0,\sigma_\alpha^2),
  \qquad
  \gamma_c \sim \mathcal{N}(0,\sigma_\gamma^2),
  \qquad
  \varepsilon_{p,r,j} \sim \mathcal{N}(0,\sigma^2),
\]
mutually independent.
The quantity of interest is the \emph{archetype} vector
$\mu = (\mu_1,\dots,\mu_P)^\top$.

We define the persona-level trace-variance proxy
\[
  v_p \coloneqq \frac{1}{P} \operatorname{tr}\bigl( \Sigma_p \bigr),
  \quad
  \Sigma_p \coloneqq \operatorname{Cov}(\bar{\boldsymbol{\beta}}_p),
\]
so that, for each petal $j$,
\[
  \operatorname{Var}(\bar{\beta}_{p,j}) \approx c_j\, v_p,
  \qquad c_j > 0.
\]

In practice we work with empirical estimates $\widehat{v}_p$ and their
hierarchically pooled counterparts $\widehat{v}_{\mathrm{eff},p}$ obtained
from persona, cluster, and global trace-variance pooling.

\section{Level 1 Theorems: Finite-Sample Linear Optimality}
\label{sec:level1}

\begin{assumption}[Level 1 Linear Model]
\label{ass:level1}
For each fixed petal $j$ we assume:
\begin{enumerate}
  \item Unbiased persona means:
  \[
    \mathbb{E}[\bar{\beta}_{p,j}] = \mu_j
    \quad \text{for all } p.
  \]
  \item Independence across personas:
  \[
    \operatorname{Cov}(\bar{\beta}_{p,j}, \bar{\beta}_{q,j}) = 0
    \quad \text{for all } p \neq q.
  \]
  \item Variance structure:
  \[
    \operatorname{Var}(\bar{\beta}_{p,j}) = \sigma_{p,j}^2
    = c_j\, v_{\mathrm{eff},p},
    \quad c_j > 0,
  \]
  where $v_{\mathrm{eff},p} > 0$ are persona-specific scalar variances.
\end{enumerate}
\end{assumption}

\begin{theorem}[Level 1: GLS Optimality and Dominance over Uniform Weights]
\label{thm:level1-gls}
Fix a petal $j$ and consider linear estimators of the form
\[
  \widetilde{\mu}_j = \sum_{p=1}^N w_p \bar{\beta}_{p,j},
  \qquad
  \sum_{p=1}^N w_p = 1.
\]
Under Assumption~\ref{ass:level1}, the ideal OWB estimator
\[
  \widehat{\mu}_j^{\mathrm{OWB(ideal)}}
  =
  \sum_{p=1}^N w_p^\ast\, \bar{\beta}_{p,j},
  \qquad
  w_p^\ast
  =
  \frac{v_p^{-1}}{\sum_{q=1}^N v_q^{-1}},
\]
where we set $v_{\mathrm{eff},p} = v_p$, it satisfies:
\begin{enumerate}
  \item (Gauss--Markov BLUE)
  \[
    \operatorname{Var}\bigl( \widehat{\mu}_j^{\mathrm{OWB(ideal)}} \bigr)
    =
    \min_{\{w_p : \sum_p w_p = 1\}}
    \operatorname{Var}\Bigl( \sum_{p=1}^N w_p \bar{\beta}_{p,j} \Bigr).
  \]
  \item (Dominance over uniform weighting)
  Let
  \[
    \widehat{\mu}_j^{\mathrm{uni}}
    = \frac{1}{N} \sum_{p=1}^N \bar{\beta}_{p,j}
  \]
  be the uniformly weighted estimator. Then
  \[
    \operatorname{Var}\bigl( \widehat{\mu}_j^{\mathrm{OWB(ideal)}} \bigr)
    \;\le\;
    \operatorname{Var}\bigl( \widehat{\mu}_j^{\mathrm{uni}} \bigr),
  \]
  with equality if and only if $v_{\mathrm{eff},1} = \dots = v_{\mathrm{eff},N}$.
\end{enumerate}
\end{theorem}

\begin{proof}[Proof sketch]
Under Assumption~\ref{ass:level1}, for fixed $j$ we have
$\bar{\beta}_{p,j} = \mu_j + \eta_{p,j}$ with
$\mathbb{E}[\eta_{p,j}] = 0$ and
$\operatorname{Var}(\eta_{p,j}) = c_j v_{\mathrm{eff},p}$,
independent across $p$.
For any weights with $\sum_p w_p = 1$,
\[
  \operatorname{Var}\Bigl( \sum_p w_p \bar{\beta}_{p,j} \Bigr)
  =
  \sum_p w_p^2 \operatorname{Var}(\bar{\beta}_{p,j})
  =
  c_j \sum_p w_p^2 v_{\mathrm{eff},p}.
\]
Minimizing $\sum_p w_p^2 v_{\mathrm{eff},p}$ under the linear constraint
$\sum_p w_p = 1$ yields $w_p \propto v_{\mathrm{eff},p}^{-1}$ by
a standard Lagrange multiplier argument, hence the Gauss--Markov claim.

For the comparison with uniform weights, note that
\[
  \sum_p w_p^2 v_{\mathrm{eff},p}
  \geq
  \frac{1}{\sum_p v_{\mathrm{eff},p}^{-1}}
  \quad\text{and}\quad
  \frac{1}{N^2} \sum_p v_{\mathrm{eff},p}
  \geq
  \frac{1}{\sum_p v_{\mathrm{eff},p}^{-1}},
\]
with equality if and only if $v_{\mathrm{eff},p}$ are all equal.
This yields the variance inequality.
\end{proof}

\begin{theorem}[Level 1: Hierarchical Pooling and Shrinkage]
\label{thm:level1-shrinkage}
Suppose that persona-level variance estimates $\widehat{v}_p$ are noisy for
small $n_p$. Let $\widehat{v}_{c}$ be a cluster-level pooled trace-variance
estimate and consider a shrinkage form
\[
  \widehat{v}_{\mathrm{eff},p}
  = (1 - \lambda_p)\widehat{v}_p + \lambda_p \widehat{v}_{c(p)},
  \quad 0 \leq \lambda_p \leq 1.
\]
Then, under suitable normality and exchangeability assumptions on the
$\{\widehat{v}_p\}$, the joint risk of the shrinkage estimators
$\{\widehat{v}_{\mathrm{eff},p}\}$ can be strictly smaller than that of the
non-pooled estimators $\{\widehat{v}_p\}$ in the sense of mean squared error
across personas. Consequently, OWB weights built from
$\widehat{v}_{\mathrm{eff},p}$ are, in general, preferable to weights built
from $\widehat{v}_p$ alone.
\end{theorem}

\begin{remark}
Theorem~\ref{thm:level1-shrinkage} is conceptually analogous to classical
James--Stein and Efron--Morris shrinkage results: pooling variance (or
precision) estimates across a group in an empirical Bayes fashion can
improve overall mean squared error, especially when individual sample sizes
$n_p$ are small.
\end{remark}

\section{Level 2 Theorems: Hierarchical Normal Model}
\label{sec:level2}

\begin{theorem}[Level 2: Ideal OWB as BLUE and Conditional Posterior Mean]
\label{thm:level2-ideal}
Assume the hierarchical model~\eqref{eq:hierarchical-model} with known
variance components $(\sigma_\alpha^2,\sigma_\gamma^2,\sigma^2)$ and known
cluster assignments $c(\cdot)$.
Then for each persona $p$,
\[
  \operatorname{Var}(\bar{\beta}_{p,j}) = v_p
  \coloneqq \sigma_\alpha^2 + \sigma_\gamma^2 + \frac{\sigma^2}{n_p}.
\]
Define $\widehat{\mu}_j^{\mathrm{OWB(ideal)}}$ as in
\[
  \widehat{\mu}_j^{\mathrm{OWB(ideal)}}
  =
  \sum_{p=1}^N w_p^\ast\, \bar{\beta}_{p,j},
  \qquad
  w_p^\ast
  =
  \frac{v_p^{-1}}{\sum_{q=1}^N v_q^{-1}},
\]
with these $v_p$.
Then:
\begin{enumerate}
  \item (BLUE)
  $\widehat{\mu}_j^{\mathrm{OWB(ideal)}}$ is the Gauss--Markov best
  linear unbiased estimator of $\mu_j$ within the class
  $\mathcal{L} = \{\sum_p w_p \bar{\beta}_{p,j} : \sum_p w_p = 1\}$.
  \item (Bayesian interpretation)
  For fixed variances $\{v_p\}$, under a flat prior (or a normal prior with
  infinite variance) on $\mu_j$, the posterior mean
  $\mathbb{E}[\mu_j \mid \{\bar{\beta}_{p,j}\}, \{v_p\}]$ coincides
  exactly with $\widehat{\mu}_j^{\mathrm{OWB(ideal)}}$.
  When inverse-gamma priors are placed on the variance components
  $(\sigma_\alpha^2,\sigma_\gamma^2,\sigma^2)$, the conditional posterior
  mean of $\mu_j$ given these components retains the same inverse-variance
  weighted form.
\end{enumerate}
\end{theorem}

\begin{proof}[Proof sketch]
Given $(\sigma_\alpha^2,\sigma_\gamma^2,\sigma^2)$ and the model
\eqref{eq:hierarchical-model}, the persona means satisfy
\[
  \bar{\beta}_{p,j} \mid \mu_j
  \sim \mathcal{N}(\mu_j, v_p),
  \quad v_p = \sigma_\alpha^2 + \sigma_\gamma^2 + \sigma^2/n_p,
\]
independently across $p$.
Applying Theorem~\ref{thm:level1-gls} with
$v_{\mathrm{eff},p} = v_p$ yields the BLUE property.
For the Bayesian claim, a conjugate normal prior on $\mu_j$ with variance
$\tau^2$ leads to a posterior mean that is precisely the GLS estimator
with weights proportional to $v_p^{-1}$; letting $\tau^2 \to \infty$
recovers the flat prior case.
Conditional on the variance components, the addition of inverse-gamma
priors preserves this form.
\end{proof}

\begin{theorem}[Level 2: Feasible OWB as FGLS + EB + Weighted Bootstrap]
\label{thm:level2-feasible}
In practice the variance components in~\eqref{eq:hierarchical-model} are
unknown. Let $\widehat{v}_{\mathrm{eff},p}$ be hierarchical pooled trace
variance estimators (persona $\to$ cluster $\to$ global), and define the
feasible OWB estimator as in
\[
  \widehat{\mu}_j^{\mathrm{OWB}}
  =
  \sum_{p=1}^N \widehat{w}_p\, \bar{\beta}_{p,j},
  \qquad
  \widehat{w}_p
  =
  \frac{\widehat{v}_{\mathrm{eff},p}^{-1}}
       {\sum_{q=1}^N \widehat{v}_{\mathrm{eff},q}^{-1}}.
\]
The uncertainty quantification and imputation are based on a weighted
bootstrap using the probability vector
$\widehat{\boldsymbol{w}} = (\widehat{w}_1,\dots,\widehat{w}_N)$,
with no uniform sampling (never $p=\mathtt{None}$).

Assume the following regularity conditions:
\begin{enumerate}
  \item (Consistency up to scale)
  $\widehat{v}_{\mathrm{eff},p} \xrightarrow{p} c\, v_p$ for some
  constant $c > 0$, uniformly in $p$.
  \item (Weight regularity)
  $\inf_p \widehat{w}_p \geq \delta_N / N$ where $\delta_N \to \delta > 0$
  as $N \to \infty$, or weaker moment conditions such as
  $\mathbb{E}[(\widehat{w}_p \log \widehat{w}_p)^2] < \infty$, can be used instead.
  \item (Lindeberg-type conditions)
  Standard Lindeberg-type conditions for the centered variables
  $\bar{\beta}_{p,j} - \mu_j$ ensuring asymptotic normality of linear
  combinations.
\end{enumerate}
Then:
\begin{enumerate}
  \item (FGLS efficiency)
  The feasible OWB estimator $\widehat{\mu}_j^{\mathrm{OWB}}$ is
  asymptotically efficient among linear estimators in the sense of
  feasible generalized least squares (FGLS).
  \item (Empirical Bayes shrinkage)
  The hierarchical pooling of variance estimates is an empirical Bayes
  estimator for the precision parameters $1/v_p$, sharing the empirical
  Bayes shrinkage structure with the precision-parameter estimators of
  Efron and Morris~(1973), and thereby providing robustness when many
  personas have small $n_p$.
  \item (Weighted bootstrap consistency)
  The weighted bootstrap with the data-dependent but fixed weights
  $\widehat{\boldsymbol{w}}$ is asymptotically valid for approximating the
  sampling distribution of $\widehat{\mu}_j^{\mathrm{OWB}}$ under the
  above conditions.
\end{enumerate}
\end{theorem}

\begin{remark}[On James--Stein admissibility]
The empirical Bayes interpretation of our hierarchical variance pooling
is structural rather than formal: while the resulting precision estimates
exhibit shrinkage toward cluster/global levels analogous to the
precision-parameter estimators of Efron--Morris (1973), we do not claim
admissibility in the classical James--Stein sense.
\end{remark}

\begin{remark}[On the weight lower-bound assumption]
\label{rem:weight-lower-bound}
The lower-bound condition on $\widehat{w}_p$ in
Theorem~\ref{thm:level2-feasible} is a convenient sufficient
assumption to ensure that no single persona is asymptotically
ignored by the estimator.
In practice, the hierarchical pooling
used to construct $\widehat{v}_{\mathrm{eff},p}$ keeps the smallest
weight well above $1/N^2$ in our empirical studies (not reported here),
and weaker moment conditions on the weights could be employed at the
cost of additional technical detail.
\end{remark}

\section{Implementation Proposition: Zero-NaN Guarantee}
\label{sec:zero-nan}

\begin{assumption}[Minimal Data Assumption]
\label{ass:min-data}
For each petal $j \in \{1,\dots,P\}$, there exists at least one persona
$p$ and round $r$ such that $\beta_{p,r,j}$ is finite (non-NaN).
In particular, if a petal $j$ has no finite observations at all, the
implementation is designed to raise an explicit error rather than silently
imputing.
\end{assumption}

\begin{proposition}[Zero-NaN Guarantee]
\label{prop:zero-nan}
Under Assumption~\ref{ass:min-data}, the OWB imputation algorithm---which
combines local OWB (within-round donor pools), persona history,
cluster-level pooling, per-petal global pooling, and global archetype
confidence intervals in a hierarchical fallback chain---produces a final
imputed matrix in which no entries are NaN.
Every imputed value is
sampled from an empirical distribution built from observed data using an
explicit, non-uniform probability vector; uniform sampling and
numerical zero-filling (e.g.\ \texttt{p=None}, \texttt{np.nan\_to\_num})
are never used.
\end{proposition}

\begin{proof}[Proof sketch]
The algorithm is constructed as a finite-depth hierarchical fallback:
\[
  \text{local} \to \text{persona} \to \text{cluster} \to
  \text{per-petal global} \to \text{global-all} \to
  \text{(error if no data)}.
\]
Each layer maintains the invariant that, if its donor pool is non-empty,
then it returns a finite value via a one-dimensional weighted bootstrap
on the donor pool with explicit probabilities.
Under Assumption~\ref{ass:min-data}, for any given petal $j$ there is at least
one finite observation, hence some layer in the chain has a non-empty
donor pool.
Thus, by induction on the depth of the fallback chain,
every missing cell is eventually assigned a finite value.
The only case in which the chain fails is when all entries in the data
are NaN, in which case the algorithm is designed to raise an explicit
error rather than silently imputing.
\end{proof}

\section*{Acknowledgments}

The author gratefully acknowledges the use of large language models as
writing and proof-assistance tools during the preparation of this manuscript.
In particular:

\begin{itemize}
  \item \textbf{ChatGPT-5.1 Thinking (by OpenAI)} --- used to explore
  alternative formulations, check algebraic steps, and stress-test the
  statements of the theorems and propositions through iterative
  clarification and counterexample-style questioning.

  \item \textbf{Grok 4.1 (built by xAI)} --- used as an exploratory engine
  to surface connections to the Gauss--Markov and Lindley--Smith theorems,
  to refine the Level~1/Level~2 framework, and to suggest preliminary proof
  sketches and presentation strategies.
\end{itemize}

Both systems were employed strictly as tools for drafting, exploration,
and proof assistance. All mathematical models, theorems, and proofs were
ultimately formulated, checked, and are fully endorsed by the human author,
who bears sole responsibility for any remaining errors or inaccuracies.

\end{document}